\documentclass[11pt]{article}
\usepackage{times}
\usepackage{enumerate}
\usepackage{comment}
\usepackage{url}
\usepackage{subfig}
\usepackage{amsmath}
\usepackage{amsthm}
\usepackage{amssymb}
\usepackage{wrapfig}
\usepackage{graphicx}
\graphicspath{{./Figs/}}
\usepackage{paralist}
\usepackage[active]{srcltx}
\usepackage{algorithm}
\usepackage{algorithmic}
\usepackage[a4paper,hmargin=1in,vmargin=1.1in]{geometry}
\usepackage{color}
\newenvironment{proof sketch}[1]{\noindent {\emph{Proof sketch of #1:}}}{\hfill \qed}
\newtheorem{theorem}{Theorem}
\newtheorem{proposition}{Proposition}
\newtheorem{lemma}{Lemma}

\newtheorem{coro}{Corollary}
\newtheorem{rem}{Remark}

\newcommand{\eqdf}{\triangleq}

\newcommand{\NNN}{[N]}

\newcommand{\dlp}{\textsc{d-lp}}
\newcommand{\plp}{\textsc{p-lp}}
\newcommand{\valp}{P}

\newcommand{\alg}{\textsc{alg}}
\newcommand{\opt}{\textsc{opt}}

\newcommand{\PR}{\Delta_t P}

\newcommand{\PRj}{\Delta_j P}
\newcommand{\PRaj}{\Delta_{a_j} P}
\newcommand{\PRbj}{\Delta_{b_j} P}

\newcommand{\PRx}[1]{\Delta_{#1} P}

\newcommand{\TS}{a_k}
\newcommand{\TE}{b_k}
\newcommand{\TSJ}{a_j}

\newcommand{\initt}{\text{\textsf{start}}}
\newcommand{\event}{\text{$\sigma$}}
\newcommand{\load}{\text{\textit{load}}}
\newcommand{\alive}{\text{\textit{Alive}}}
\newcommand{\dead}{\text{\textit{Dead}}}
\newcommand{\cut}{\text{\textit{cut}}}
\newcommand{\paths}{\text{\textit{paths}}_t}
\newcommand{\pathsp}{\text{\textit{paths}}_{t+1}}
\newcommand{\pathtt}[1]{\text{\textit{paths}}_{#1}}

\newcommand{\mainn}{\text{\textsc{Main}}}
\newcommand{\routealg}{\text{\textsc{Route}}}
\newcommand{\unroutealg}{\text{\textsc{UnRoute}}}
\newcommand{\depart}{\text{\textsc{Depart}}}
\newcommand{\reroutealg}{\text{\textsc{MakeFeasible}}}

\newcommand{\NN}{{\mathbb{N}}}

\begin{document}

\title{A Nonmonotone Analysis with the Primal-Dual Approach: online routing of virtual circuits with unknown durations}

\author{%
Guy Even\thanks{School of Electrical Engineering, Tel-Aviv
Univ., Tel-Aviv 69978, Israel.
\protect\url{{guy,medinamo}@eng.tau.ac.il}.\newline
M.M was partially funded
by the Israeli Ministry of Science and Technology.} \and Moti Medina$^*$
}
\date{}
\maketitle

\begin{abstract}
  We address the question of whether the primal-dual approach for the
  design and analysis of online algorithms can be applied to
  nonmonotone problems. We provide a positive answer by presenting a
  primal-dual analysis to the online algorithm of Awerbuch
  et al.~\cite{awerbuch2001competitive} for routing virtual circuits with unknown durations.
%
\end{abstract}


\section{Introduction}
The analysis of most online algorithms is based on a potential
function (see, for example,~\cite{AAP, azar1997line, aspnes1997line, awerbuch2001competitive} in the context of online routing).
Buchbinder and Naor~\cite{BNsurvey} presented a primal-dual approach for
analyzing online algorithms.  This approach replaces the need to find
the appropriate potential function by the task of finding an
appropriate linear programming formulation.

The primal-dual approach presented by Buchbinder and Naor has a
monotone nature.  Monotonicity means that: (1)~Variables and
constraints arrive in an online fashion. Once a variable or constraint
appears, it is never deleted. (2)~Values of variables, if updated, are
only increased.  We address the question of whether the primal-dual
approach can be extended to analyze nonmonotone
algorithms\footnote{The only instance we are aware of in which the
  primal-dual approach is applied to nonmonotone variables appears
  in~\cite{buchbinder2011frequency}. In this instance, the change in
  the dual profit, in each round, is at least a constant times the change in the
  primal profit. In general, this property does not hold in a nonmonotone setting.}.

An elegant example of nonmonotone behavior occurs in the problem of
online routing of virtual circuits with unknown durations. In the
problem of routing virtual circuits, we are given a graph with edge
capacities. Each request $r_i$ consists of a
source-destination pair $(s_i,t_i)$.  A request $r_i$ is served by
allocating to it a path from $s_i$ to $r_i$. The goal is to serve the
requests while respecting the edge capacities as much as possible. In
the online setting, requests arrive one-by-one. Upon arrival of a
request $r_i$, the online algorithm must serve $r_i$. In the special
case of unknown durations, at each time step, the adversary may
introduce a new request or it may terminate an existing request.  When
a request terminates, it frees the path that was allocated to it, thus
reducing the congestion along the edges in the path.  The online
algorithm has no knowledge of the future; namely, no information about
future requests and no information about when existing requests will
end. Nonmonotonicity is expressed in this online problem in two ways:
(1)~Requests terminate thus deleting the demand to serve them. (2)~The
congestion of edges varies in a nonmonotone fashion; an addition of a
path increases congestion, and a deletion of a path decreases
congestion.

Awerbuch et al.~\cite{awerbuch2001competitive} presented an online
algorithm for online routing of virtual circuits when the requests
have unknown durations. In fact, their algorithm resorts to rerouting
to obtain a logarithmic competitive ratio for the load. Rerouting means
that the path allocated to a request is not fixed and the algorithm
may change this path from time to time. Hence, allowing rerouting
increases the nonmonotone characteristics of the problem.

We present an analysis of the online algorithm of Awerbuch et
al.~\cite{awerbuch2001competitive} for online routing of virtual circuits with
unknown durations. Our analysis uses the primal-dual approach, and hence
we show that the primal-dual approach can be applied in nonmonotone
settings.

\section{Problem Definition}
\subsection{Online Routing of Virtual Circuits with Unknown Durations}
Let $G=(V,E)$ denote a directed or undirected graph.
Each edge $e$ in $E$ has a capacity $c_e\geq 1$.  A routing request $r_k$ is a $4$-tuple $r_k=
(s_k,d_k,\TS,\TE)$, where
\begin{inparaenum}[(i)]
\item
  $s_k,d_k \in V$ are the source and the destination of the $k$th
  routing request, respectively,
\item $\TS\in \NN$ is both the arrival time and the start time of the request, and
\item $\TE \in \NN$ is the departure time or end time of the request.
\end{inparaenum}
Let $\Gamma_k$ denote the set of paths in $G$ from $s_k$ to $d_k$.
A request $r_k$ is served if it is allocated  a path in $\Gamma_k$.

Let $\NNN$ denote the set $\{0,\ldots,N\}$.
The input consists of a sequence of events
$\sigma=\{\event_t\}_{t \in \NNN}$.  We assume that time is discrete, and
event $\event_t$ occurs at time $t$.
There are two types of events:
\begin{inparaenum}[(i)]
\item An \emph{arrival} of a request. When a request $r_k$ arrives, we
  are given the source $s_k$ and the destination $d_k$. Note that the
  arrival time $\TS$ simply equals the current time $t$.
\item A \emph{departure} of a request. When a request $r_k$ departs
  there is no need to serve it anymore (namely, the departure time $\TE$ simply equals the  current time $t$).
\end{inparaenum}

The set of active requests at time $t$ is denoted by $\alive_t$ and is defined by
\[
    \alive_t \eqdf \{ r_k \mid \TS\lneqq t \leq \TE\}\:.
\]

An \emph{allocation} is a sequence $A=\{p_k\}_k$ of paths such that
$p_k$ is a path from the source $s_k$ to the destination $d_k$ of request $r_k$.
Let $\paths(e,A)$ denote the number of requests that are routed along edge $e$ by allocation $A$ at time $t$, formally:
\[
\paths(e,A) \eqdf \left|\{p_k : e\in p_k \text{ and } r_k \in \alive_t\}\right|\:.
\]
The \emph{load} of an edge $e$ at text $t$ is defined by
\[
\load_t(e,A) \eqdf \frac{\paths(e,A)}{c_e}\:.
\]
The \emph{load} of an allocation $A$ at time $t$ is defined by
\[
\load_t(A) \eqdf \max_{e\in E} \load_t(e,A) \:.
\]
The \emph{load} of an allocation $A$ is defined by
\[
\load(A) \eqdf \max_t \load_t (A)\:.
\]

An algorithm computes an allocation of paths to the requests, and
therefore we abuse notation and identify the algorithm with the
allocation that is computed by it. Namely, $\alg(\sigma)$ denotes the
allocation computed by algorithm $\alg$ for an input sequence
$\event$.

In the online setting, the events arrive one-by-one, and no information is known about an event before its arrival.
Moreover,
\begin{inparaenum}[(1)]
  \item the length $N$ of the sequence of events is unknown; the input simply stops at some point,
  \item the departure time $\TE$ is \emph{unknown} (and may even be determined later by the adversary),  and
  \item the online algorithm must allocate a path to the request as soon as the request arrives.
\end{inparaenum}

The \emph{competitive ratio} of an online algorithm $\alg$ with respect to $N \in \NN$, and a sequence $\event=\{\event_t\}_{t \in \NNN}$ is defined
by
\[
\rho(\alg(\sigma)) \eqdf \frac{\load(\alg(\sigma))}{\load(\opt(\sigma))} \:,
\]
where $\opt(\sigma)$ is an allocation with minimum load.
The \emph{competitive ratio} of an online algorithm $\alg$ is defined by
\[
\rho(\alg) \eqdf \sup_{N \in \NN}~\max_{\event} \rho(\alg(\sigma)) \:.
\]

\noindent
Note that since every request has a unit demand, we may assume that $c_e \geq 1$ for every edge $e \in E$.
\subsection{Rerouting}
In the classical setting, a request $r_k$ is served by a fixed single
path $p_k$ throughout the duration of the request.  The term \emph{rerouting} means that we allow the allocation to change the path $p_k$ that serves $r_k$.  Thus, there are two extreme cases: (i)~no rerouting at all is permitted (classical
setting), and (ii)~total
flexibility in which, a new allocation can be computed in each time step.

Following the paper by Awerbuch et al.~\cite{awerbuch2001competitive}, we allow the online algorithm to reroute each request
at most $O(\log |V|)$ times. In the analysis of the competitive ratio,
we compare the load of the online algorithm with the load of an
optimal (splittable) allocation with total rerouting flexibility. Namely, the
optimal solution recomputes a minimum load allocation at each time step, and, in addition may serve a request by a convex combination of paths.

\section{The Online Algorithm \alg}\label{sec:gen}
In this section we present the online algorithm \alg\ that is listed in Algorithm~\ref{alg:alg}.
Thus algorithm is equivalent to the algorithm presented in~\cite{awerbuch2001competitive}.

\medskip
\noindent
The algorithm maintains the following variables.
\begin{enumerate}
  \item For every edge $e$ a variable $x_e$. The value of $x_e$ is exponential in the load of edge $e$.
  \item For every request $r_k$ a  variable $z_k$. The value of $z_k$ is the complement of the ``weight'' of the path $p_k$ allocated to $r_k$ at the time the path was allocated.
  \item For every routing request $r_k$, and for every path $p \in \Gamma_k$ a variable $f_k(p)$. The value of $f_k(p)$ indicates whether $p$ is allocated to $r_k$. That is, the value of $f_k(p)$ equals $1$ if path $p$ is allocated for request $r_k$, and $0$ otherwise.
\end{enumerate}

The algorithm \alg\ consists of the following $5$ procedures:
(1)~\mainn, (2)~\routealg, (3)~\depart, (4)~\unroutealg, and (5)~\reroutealg.

The \mainn\ procedure begins with initialization.
For every $e \in E$, $x_e$ is initialized to $\frac{1}{4m}$, where $m=|E|$.
For every $k \in \NNN$, $z_k$ is initialized to zero.
For every $k \in \NNN$, and for every path $p$, $f_k(p)$ is  initialized to zero.
Since the number of $z_k$ and $f_k(p)$ variables is unbounded, their initialization is done in a ``lazy'' fashion; that is, upon arrival of the $k$th request the corresponding variables are set to zero.

The main procedure \mainn\ proceeds as follows.
For every time step $t \in \NNN$ , if
the event $\event_t$ is an arrival of a request, then the \routealg\ procedure is invoked.
Otherwise, if
the event $\event_t$ is a departure of a request, then the \depart\ procedure is invoked.

The \routealg\ procedure serves request $r_k$ by allocating a
``lightest'' path $p_k$ in the set $\Gamma_k$ (recall that $\Gamma_k$ denotes
the set of paths from the source $s_k$ to the destination $d_k$). The
allocation is done by two actions. First, the allocation of $p_k$ to request
$r_k$ is indicated by setting $f_k(p_k) \leftarrow 1$. Second, the loads of the edges along $p_k$ are updated by increasing the
variables $x_e$ for $e \in p_k$.  The variable $z_k$ equals the
``complement'' weight of the allocated path $p_k$.  Note that this
complement is with respect to half the weight of the path before its
update.

The \depart\ procedure ``frees'' the path that is allocated for $p_k$,
by calling the \unroutealg\ procedure.  The \unroutealg\ procedure
frees $p_k$ by nullifying $f_k(p_k)$ and $z_k$, and by decreasing the
edge variables $x_e$ for the edges along $p_k$.  The freeing of $p_k$
decreases the load along the edges in $p_k$. As a result of this
decrease, it may happen that a path allocated
to an alive request might be very heavy compared to a lightest path.
In such a case, the request should be rerouted. This is why the
 \reroutealg\ procedure is invoked after the \unroutealg\ procedure.

Rerouting is done by the \reroutealg\ procedure. This rerouting is done by freeing a path and then routing the request again.
Requests with improved alternative paths are rerouted.

The listing of the online algorithm \alg\ appears in Algorithm~\ref{alg:alg}.

\begin{algorithm}
\footnotesize
    \underline{\textbf{\mainn(}$\event_t$)}
    \begin{algorithmic}[1]
     \STATE $\forall k\in \NNN : z_k \leftarrow 0$.
     \STATE$\forall e\in E : x_e \leftarrow \frac{1}{4m}$, where $m=|E|$.
     \STATE$\forall r_k \in \NNN ~\forall p : f_k(p) \leftarrow 0$.
     \STATE \textbf{Upon } {arrival of event $\event_t$ }\textbf{do}
       \STATE \textbf{\quad if} $\event_t$ is an arrival of request $r_k$ \textbf{then}
        \textbf{Call \routealg($r_k$)}. \label{line:main5}
       \STATE \textbf{\quad else} ($\event_t$ is an departure of request $r_k$)
        \textbf{Call \depart($r_k$)}.
    \end{algorithmic}

    \smallskip
    \underline{$\textbf{\routealg($r_k$)}$}
    \begin{algorithmic}[1]
        \STATE \label{state:min path} Find the ``lightest'' path: $ p_k \leftarrow \text{argmin}\{\sum_{e\in p'}\frac{x_{e}}{c_e} \mid p' \in \Gamma_k\}$.
        \STATE \label{state:route z}$z_k \leftarrow 1- \frac 12 \cdot \sum_{e\in p_k}\frac{x_{e}}{c_e}$.
        \STATE Route $r_k$ along $p_k$: $f_k(p_k) \leftarrow 1$.
        \STATE \textbf{for all} $e\in p_k$ \textbf{do}
            \STATE \quad \label{state:route x}$x_{e} \leftarrow x_{e}\cdot \lambda_e$ where $\lambda_e \eqdf \left( 1 + \frac{1}{4c_e} \right)$.
            \COMMENT {Update edge ``load''}
    \end{algorithmic}

    \smallskip
    \underline{$\textbf{\depart($r_k$)}$}
    \begin{algorithmic}[1]
        \STATE \textbf{Call} \unroutealg($r_k$).
        \STATE \textbf{Call} \reroutealg($x,z$).
    \end{algorithmic}

    \smallskip
    \underline{$\textbf{\unroutealg($r_k$)}$}
    \begin{algorithmic}[1]
        \STATE Free variables: $z_k , f_k(p_k)$.
        \STATE \textbf{for all} $e\in p_k$ \textbf{do}
            \STATE \quad \label{state:unroute x}$x_{e}\leftarrow x_{e}/\lambda_e$ where $\lambda_e\eqdf \left( 1 + \frac{1}{4c_e} \right)$.
            \COMMENT {Update edge ``load''}
    \end{algorithmic}

    \smallskip
    \underline{$\textbf{\reroutealg($x,z$)}$}
    \begin{algorithmic}[1]
        \STATE {$\forall r_j \in \alive_t$ \textbf{if} $~\exists p \in \Gamma_j~: z_j + \sum_{e\in p}\frac{x_{e}}{c_e} < 1$ \textbf{then}}
            \STATE \textbf{\quad Call} \unroutealg($r_j$).
            \STATE \textbf{\quad Call} \routealg($r_j$). \label{line:makefeas3}
    \end{algorithmic}
\caption{\alg: Online routing algorithm. The input consists of (1)~ a graph $G=(V,E)$ where each $e\in E$ has capacity $c_e$, and (2)~a sequence of events $\sigma = \{\event_t\}_{t\in\NNN}$.
}\label{alg:alg}
\end{algorithm}

\section{Primal-Dual Analysis of \alg}~\label{sec:analysis} In this section we prove that the load on
every edge is always $O(\log |V|)$, and that each request is rerouted
at most $O(\log |V|)$ times.  We refer to an input sequence $\event$
as \emph{feasible} if there is an allocation $A$, such that for all
requests that are alive at time $t$, it holds that $\load_t(A) \leq
1$.  The following theorem holds under the assumption that the input
sequence $\event$ is feasible. Note that the removal of this
assumption increases the competitive ratio only by a constant factor
by standard doubling techniques~\cite{awerbuch2001competitive}.

\begin{theorem}[\cite{awerbuch2001competitive}]~\label{thm:main result}
If the input sequence $\event$ is feasible and assuming that $c_e \geq 1$, then \alg\ is:
 \begin{enumerate}
   \item An $O(\log |V|)$-competitive online algorithm.
   \item Every request is rerouted at most $O(\log |V|)$ times.
 \end{enumerate}
\end{theorem}

We point out that the allocation computed by $\alg$ is
\emph{nonsplittable} in the sense that at every given time each
request is served by a single path.  The optimal allocation, on the
other hand, is both totally
flexible and \emph{splittable}. Namely, the optimal allocation may reroute all
the requests in each time step, and, in addition, may serve a request
by a convex combination of paths.

\medskip \noindent
The rest of the proof is as follows.
We begin by formulating a packing and covering programs for our problem in Section~\ref{sec:lp}.
We then prove Lemma~\ref{lemma:x ub} in Section~\ref{sec:proofof5}.
We conclude the analysis with the proof of Theorem~\ref{thm:main result} in Section~\ref{sec:main}
\subsection{Formulation as an Online Packing Problem}\label{sec:lp}
For the sake of analysis, we define for every prefix of events $\{\event_j\}_{j=1}^t$ a \emph{primal} linear program $\plp(t)$ and its \emph{dual} linear program $\dlp(t)$. The primal LP is a \emph{covering} LP, and the dual LP is a \emph{packing} LP.
The LP's appear in Figure~\ref{fig:LPframe}.

The variables of the LPs correspond to the variables maintained by
\alg, as follows.  The covering program $\plp(t)$ has a variable $x_e$
for every edge $e \in E$, and a variable $z_k$ for every $r_k \in
\alive_t$.  The packing program $\dlp(t)$ has a variable $f_k(p)$ for
every request $r_k \in \alive_t$, and for every path $p \in \Gamma_k$.  The
variable $f_k(p)$ equals to the fraction of $r_k$'s ``demand'' that is
routed along path $p \in \Gamma_k$.

The dual LP has three types of constraints: capacity constraints,
demand constrains, and sign constraints.  In the fractional setting
the load of an edge is defined by
\[
\load_t(e) \eqdf \frac{1}{c_e} \cdot
\sum_{r_k \in \alive_t} \sum_{\{p \mid p\in \Gamma_k, e \in p\}}f_k(p)
\:.
\]
The capacity constraint in the dual LP requires that the load of each
edge is at most one.  The demand constraints require that each request
$r_k$ that is alive at time $t$ is allocated a convex combination of
paths.

If the dual LP is feasible, then the
objective function of the dual LP simply equals the number of requests
that are alive at time step $t$, i.e., $|\alive_t|$.

The primal LP has two types of constraints: covering constraints and
sign constraints.  The covering constraints requires that for every
request $r_k$ that is alive and for every path $p \in \Gamma_k$, the sum of
$z_k$ and the ``weight'' of $p$ is at least~$1$.  Note that the sign
constraints apply only to the edge variables $x_e$ whereas the request
variables $z_k$ are free.

\begin{figure}
\centering
\begin{tabular}{| l |}
\hline
    \begin{minipage}{0.67\textwidth}
     \begin{center}
        \begin{eqnarray*}
        \underline{\plp(t)}:~~ \min~  \sum_{r_k \in \alive_t} z_k +\sum_{e\in E} x_{e} &\text{s.t.}& \\
         \forall r_k\in \alive_t ~\forall p \in \Gamma_k: z_k + \sum_{e\in p}\frac{x_{e}}{c_e}  &\geq& 1\text{ \footnotesize (Covering Constraints.)}\\
          x&\geq& \vec{0}
        \end{eqnarray*} \\
        (I)
     \end{center}
    \end{minipage}
\\
\hline
    \begin{minipage}{0.67\textwidth}
     \begin{center}
        \begin{eqnarray*}
         \underline{\dlp(t)}:~~~~~\max~ \sum_{r_k \in \alive_t} \sum_{p \in \Gamma_k} f_k(p) &\text{s.t.}&\\
         \forall e\in E: \frac{1}{c_e} \cdot \sum_{r_k \in \alive_t} \sum_{\{p \mid p\in \Gamma_k, e \in p\}}f_k(p)&\leq& 1\text{ \footnotesize (Capacity Constraints.)} \\
         \forall r_k\in \alive_t: \sum_{p \in \Gamma_k}f_k(p) & = & 1\text{ \footnotesize (Demand Constraints.)}\\ f &\geq& \vec{0}
        \end{eqnarray*}
        \\(II)
     \end{center}
    \end{minipage}
    \\
\hline
\end{tabular}
\caption{
(I) The primal LP, $\plp(t)$.
(II) The dual LP, $\dlp(t)$.}
   \label{fig:LPframe}
\end{figure}

Note that the assumption that $\event$ is feasible is equivalent to requiring that
the dual program $\dlp(t)$ is feasible for every $t$.

\subsection{Bounding the Primal Variables}\label{sec:proof}\label{sec:proofof5}
In this section we prove  that the primal variables $x_e$ are bounded by a constant, as formalized in the following Lemma.
 \begin{lemma}\label{lemma:x ub}
If $\sigma_t$ is an original event, then
  $$\forall e \in E~:x_{e}^{(t)} \leq 3\:.$$
\end{lemma}
\noindent
The proof of Lemma~\ref{lemma:x ub} is based on a few lemmas that we prove first.

\medskip
\paragraph{Notation.}
Let $x^{(t)}_e,z_k^{(t)}$ denote the value of the primal variables $x_e, z_k$ before event $\event_t$ is processed by \alg.
Let $\valp_t$ denote the objective function's value of $\plp(t)$, formally:
\[
    \valp_t \eqdf \sum_{r_k \in \alive_t} z_k^{(t)} +\sum_{e\in E}x_e^{(t)}\:.
\]
Let $\PR \eqdf \valp_{t+1} - \valp_{t}$.

Note that $P_t$ refers to the value of $\plp(t)$ at the beginning of
time step $t$.  The definition of $\alive_t$ implies that the
constraints and variables of $\plp(t)$ are not influenced by the event
$\sigma_t$ (this happens only for $\plp(t+1)$). Hence the variables in
the definition of $\valp_t$ are indexed by time step $t$.

\paragraph{Dummy events.}
The procedure \routealg\ is invoked in two places: (i)~in
Line~\ref{line:main5} of \mainn\ as a result of an arrival of a
request, or (ii)~in Line~\ref{line:makefeas3} of \reroutealg.  To
simplify the discussion, we create ``dummy'' events each time the
\reroutealg\ procedure reroutes a request.  Dummy events come in
pairs: first a dummy departure event for request $r_k$ is introduced,
and then a dummy arrival event for a ``continuation'' request $r_k$ is introduced.  The
combination of original events and dummy events describes the
execution of \alg.  The augmentation of the original input sequence of
events by dummy events does not modify the optimal value of the dual
LP at time steps $t$ that correspond to original events.  Hence, we
analyze the competitive ratio $\rho(\alg(\sigma))$ by analyzing the
competitive ratio with respect to the augmented sequence at time steps
$t$ that correspond to original events.

\medskip
The following lemma follows immediately from the description of the
algorithm \alg\ and the definition of dummy events.
\begin{lemma}[\textbf{Primal Feasibility}]\label{lemma:2x feas}\label{lemma:stability}
  If $\sigma_t$  is an
  original event, then the variables $\{x_e^{(t)}\}_{e \in E} \cup
  \{z_\ell^{(t)}\}_{\ell \in \alive_t}$ constitute a feasible
  solution for $\plp(t)$.
\end{lemma}
\begin{proof}
  When an original event $\sigma_{t'}$ occurs, the \reroutealg\
  procedure generates dummy events at the end of the time step to
  guarantee that the primal variables are a feasible solution of the
  primal LP.  Hence, if $\sigma_t$ is an original event, then the
  primal variables at the beginning of time step $t$ are a feasible
  solution for $\plp(t)$.
\end{proof}

\begin{lemma}\label{delta primal arrival}
If $\sigma_t$ is an arrival of request, then $\PR < 1$.
\end{lemma}
\begin{proof}
  Assume that $\sigma_t$ is an event in which request $r_k$ arrives.
  In Step~\ref{state:route z} of the \routealg\ algorithm $z_k$ is set
  to $1-\frac 12 \cdot \sum_{e\in p_k}\frac{x_{e}^{(t)}}{c_e}$.  In
  Step~\ref{state:route x} of the \routealg\ algorithm, for every $e
  \in p_k$, $x_e$ is increased by $\frac{x_e^{(t)}}{4c_e}$.  All the
  other edge variables $x_e$ remain unchanged.  Hence,
    \begin{align}
      \PR =&  1- \frac 12 \cdot \sum_{e\in p_k}\frac{x_{e}^{(t)}}{c_e} + \sum_{e\in p_k} \frac{x_e^{(t)}}{4c_e} \nonumber\\
          =&1- \frac 14 \cdot \sum_{e\in p_k} \frac{x_e^{(t)}}{c_e}  \label{eqn:case 1}\\
          < & 1\:, \nonumber
    \end{align}
    as required.
\end{proof}

We refer to the number of requests that are routed along edge $e$ by allocation $\alg$ at time $t$ by $\paths(e)$.
\begin{lemma}\label{lemma:valx}
For every $t$ and $e \in E$,
    \begin{align*}
      x_e ^{(t)} = \frac{1}{4m} \cdot \lambda_e^{\paths(e)}\:.
    \end{align*}
\end{lemma}
\begin{proof}
  The proof is by induction on $t$. At time $t=0$, we have
  $x_e^{(0)}=\frac{1}{4m}$ and $\paths(e)=0$.  The proof of the
  induction basis for $t+1$ depends on whether at time step $t$ an
  arrival or a departure occurs.  If the event does not affect edge
  $e$, then the induction step clearly holds. Assume that the event
  affects edge $e$. If a request $r_k$ arrives at time $t$, then
  $\pathsp(e)=\paths(e)+1$ and $x_e^{(t+1)} = x_e^{(t)} \cdot \lambda_e$.
If a request $r_k$ departs at time $t$, then
  $\pathsp(e)=\paths(e)-1$ and $x_e^{(t+1)} = x_e^{(t)}/\lambda_e$.
\end{proof}

Let $\dead_t \eqdf \{r_k \mid b_k < t\}$.
In general, it is not true that $\PRaj+\PRbj \leq 0$, however on average it is true, as stated in the following lemma.

\begin{lemma}\label{eqn:np}
    For every $t$,
    \begin{align}
      \sum_{r_j \in \dead_t}\left(\PRaj +\PRbj \right) \leq 0\:. \label{eqn:canceling}
    \end{align}
\end{lemma}
\begin{proof}
    First we prove the following proposition.
    \begin{proposition}\label{prop:intervals}
      Consider a set of $I=\{I_j=[\alpha_j,\beta_j]\}_{j=1}^q$ such that no
      two intervals share a common endpoint.  Let $\cut(t)$ denote the number of intervals that contain $t$.
    Then, there is a  permutation $\pi:[1,q] \rightarrow [1,q]$ such that
    \begin{align}
        \forall j\in [1,q]~:~\cut(\alpha_j) = \cut(\beta_{\pi(j)})\:.\label{eqn:mapping}
    \end{align}
    \end{proposition}
    \begin{proof}
      The proof is by induction on the number of intervals.  The induction basis,
      for $q=1$ holds trivially because $\cut(\alpha_1) = \cut(\beta_1)=1$.
      The proof of the induction step is based on the existence of a
      pair $\alpha_i<\beta_j$ such that the open interval $(\alpha_i,\beta_j)$ does not
      contain any endpoint of the intervals in $I$. For such a pair,
      we immediately have $\cut(\alpha_i)=\cut(\beta_j)$ so we define
      $\pi(i)=j$ and apply the induction hypothesis.

      We first show that such a pair $\alpha_i<\beta_j$ exists.  We say that an
      interval $I_m$ is \emph{minimal} if $I_m \cap I_k \neq
      \emptyset$ implies that $I_m \subseteq I_k$.  If there exists a
      minimal interval $I_m$, then set $\alpha_i=\alpha_{m}$ and $\beta_j=\beta_{m}$. In
      such a case since $\pi(m)=m$, we can erase $I_m$ and proceed by
      applying the induction hypothesis to the remaining intervals. Note that equality of cut sizes is preserved when the interval $I_m$ is deleted.

      Consider the set of pairs of
      intersecting intervals without containment defined as follows
\[
A \eqdf \{ (i,j) \mid \alpha_j<\alpha_i<\beta_j<\beta_i\}\:.
\]
If there is no minimal interval, the set $A$ is not empty.
Any pair $(i,j)\in A$ that minimizes the difference $(\beta_j-\alpha_i)$ has the property
that the interval $(\alpha_i,\beta_j)$ lacks endpoints of intervals in $I$.

We can define $\pi(i)=j$. We proceed by applying the induction
hypothesis on $\left(I \setminus \{I_j,I_i\}\right)\cup I_{k}$, where
$I_{k} = I_{i}\cup I_{j}$. Note that equality of cut sizes is preserved when $I_i$ and $I_j$ are merged into one interval.
\end{proof}

The difference $\PRx{a_j}$ consists of two parts:
\[
\PRx{a_j} = z_j^{(a_j+1)} +\sum_{e\in p_j} \frac{x_e^{(a_j)}}{4c_e}.
\]
The difference $\PRx{b_j}$ consists of two parts as well:
\[
\PRx{b_j} = -z_j^{(b_j)} -\sum_{e\in p_j} \frac{x_e^{(b_j+1)}}{4c_e}.
\]

It follows that
\begin{align*}
       \sum_{r_j \in \dead_t}\left(\PRaj +\PRbj \right) =&\sum_{r_j \in \dead_t} \sum_e \frac{1}{4c_e}\cdot \left(x_e^{(a_j)} - x_e^{(b_j+1)}\right)\\
=&
\sum_{r_j \in \dead_t} \sum_e \frac{1}{4c_e}\cdot \left(x_e^{(a_j)} - x_e^{(b_{\pi(j)}+1)}\right),\\
\end{align*}
where $\pi$ is any permutation over the set of requests.  In fact, we
shall use for each edge $e$, a different permutation $\pi=\pi(e)$ that
is a permutation over the requests $r_k$ such that $e\in p_k$.

Assume first that $\alive_t =\emptyset$. We later lift this
assumption.

Fix an edge $e$.  For each request $r_j$ such that $e\in p_j$, map the
duration $(a_j,b_j]$ of request $r_j$ to the interval $[a_j+1,b_j]$.
The resulting set of intervals satisfies $\cut(t)=\paths(e)$ for every time step $t$.  Let
$\pi$ denote the permutation guaranteed by Prop.~\ref{prop:intervals}.
Then, it suffices to prove that
\begin{align}\label{eq:local}
x_e^{(a_j)} - x_e^{(b_{\pi(j)}+1)} =0.
\end{align}
Indeed, by Lemma~\ref{lemma:valx},
$4m \cdot\left(x_e^{(a_j)} - x_e^{(b_{\pi(j)}+1)}\right) = \lambda_e^{\pathtt{a_j}}-\lambda_e^{\pathtt{b_{\pi(j)}+1}}$.
In addition, the property of permutation $\pi$ states that $\cut(a_j+1)=\cut(b_{\pi(j)})$.
It follows that $\pathtt{a_j+1}=\pathtt{b_{\pi(j)}}$.
But, $\pathtt{a_j} =\pathtt{a_j+1}-1$ and
$\pathtt{b_{\pi(j)}+1} = \pathtt{b_{\pi(j)}}-1$, and Equation~\ref{eq:local} follows.

To complete the proof, consider the requests in $\alive_t$.
Because $a_j, b_{\pi(j)}\leq t$, requests in $\alive_t$ do not increase the difference $x_e^{(a_j)} - x_e^{(b_{\pi(j)}+1)}$. Thus
$x_e^{(a_j)} - x_e^{(b_{\pi(j)}+1)}\leq 0$, and the lemma follows.
\end{proof}

We are now ready to prove Lemma~\ref{lemma:x ub}.  Recall that
Lemma~\ref{lemma:x ub} states that the primal variables $x_e$ are
bounded by a constant. The proof of Lemma~\ref{lemma:x ub} is by
contradiction. In fact, we reach a contradiction to \emph{weak
  duality}, that is, we show that the value of the primal solution is
strictly smaller than the value of a feasible dual solution.

\begin{proof}[Proof of Lemma~\ref{lemma:x ub}]
  The proof is by contradiction.  Assume $x_e^{(t)}>3$ and $\sigma_t$ is an original event. Define
\[
t_2 \eqdf \min\{t \mid
x_e^{(t)} >3 \text{ and $\sigma_t$ is an original event}\}.
\]
Let $t_1$ be the time step for which $x_e^{(t_1)}<1$ and $x_e^{t'}
\geq 1$ for every $t'\in [t_1+1,t_2]$.

Define:
\begin{align*}
\alive_{\in e}(t_1,t_2) &\eqdf \{r_j \mid t_1 <  a_j < t_2 < b_j, e \in p_j\}.
\end{align*}

Let $\delta_e$ denote
the difference between the number of arrivals and the number of
departures in the time interval $[t_1,t_2)$ among the requests that
were routed along $e$.
Clearly $\delta_e \leq |\alive_{\in e}(t_1,t_2)|$.

    Lemma~\ref{lemma:valx} implies that
    \begin{align}
      x_e^{(t_2)} =& x_e^{(t_1)}\cdot \left(1+
        \frac{1}{4c_e}\right)^{\delta_e}. \nonumber \label{eqn:gneq2}
    \end{align}
The assumption that $x_e^{(t_2)} >3$ and
$x_e^{(t_1)}<1$ imply
      \[
        \left(1+\frac{1}{4c_e}\right)^{\delta_e} \geq 3\:.
      \]
Since  $1+x \leq e^x$, it follows that $\delta_e > 4\cdot c_e$.
Hence,
    \begin{align}
       |\alive_{\in e}(t_1,t_2)| > 4\cdot c_e \label{eqn:bigdelta}\:.
    \end{align}

By Equation~\ref{eqn:case 1},
for each $r_j\in \alive_{\in e}(t_1,t_2)$, we have:
\begin{align}\label{eq:Praj}
  \PRaj<1-\frac{1}{4c_e}.
\end{align}
Hence,
    \begin{eqnarray}
        \valp_{t_2} & = & \frac {1}{4m}\cdot m + \sum_{t=0}^{t_2-1} \PR  \nonumber\\
        & = & \frac 14+ \sum_{r_j \in \dead_{t_2}}( \PRaj +\PRbj) + \sum_{r_j \in \alive_{t_2}} \PRaj \nonumber\\
        & \leq & \frac 14+  \sum_{r_j \in \alive_{t_2}} \PRaj \nonumber\\
        & < & \frac 14+
        |\alive_{t_2}|
-
        \frac{|\alive_{\in e}(t_1,t_2)| }{4c_e}
\nonumber\\
        & < &|\alive_{t_2}| \label{eqn: finaly}\:.
    \end{eqnarray}
    The justification for these lines is as follows.
    The first line follows from the initialization of the primal variables.
    The second line follows since every event in time step $t\in[0,t_2-1]$ is either
an arrival of a request in $\dead_{t_2}\cup \alive_{t_2}$ or a departure of a request in $\dead_{t_2}$.
    The third inequality is due to Lemma~\ref{eqn:np}.
    The fourth equation is due to Equation~\ref{eq:Praj}.
    The last inequality follows from Equation~\ref{eqn:bigdelta}.

    By Lemma~\ref{lemma:stability}, the primal variables at time $t_2$
    are a feasible solution of $\plp(t_2)$.  The optimal value of
    $\dlp(t_2)$ equals $|\alive_{t_2}|$.  Hence, Equation~\ref{eqn:
      finaly} contradicts weak duality, and the lemma follows.
\end{proof}
\subsection{Proof of Theorem~\ref{thm:main result}}\label{sec:main}
We now turn to the proof of the main result.
The proof is as follows.

\begin{proof}[Proof of Theorem~\ref{thm:main result}]
We begin by proving the bound on the competitive ratio.
    Lemma~\ref{lemma:valx} states that
    \[
        \forall t  ~~\forall e\in E : x_e = \frac{1}{4m}\cdot \left(1+\frac{1}{4c_e}\right)^{\paths(e)}\:.
    \]
    Hence, by Lemma~\ref{lemma:x ub}, for each original event $\sigma_t$,
    \[
        \forall e\in E : \frac{1}{4m}\cdot \left(1+\frac{1}{4c_e}\right)^{\paths(e)} \leq 3\:.
    \]
    Since $2^x \leq 1 + x$ for all $x \in [0,1]$, it follows that for each original event $\sigma_t$
    \[
\forall e\in E : \paths(e) \leq c_e \cdot 4 \log(12m)\:,
    \]
    and the first part of the theorem follows.

\medskip
\noindent
We now prove the bound on the number of reroutes.  Rerouting an alive
request $r_j$ occurs if there exists a path $p \in \Gamma_j$ such that
$\sum_{e\in p}\frac{x_{e}}{c_e} < 1-z_j$.  By Line~\ref{state:route z}
of the \routealg\ algorithm, this condition is equivalent to:
$\sum_{e\in p}\frac{x_{e}}{c_e} < \frac{1}{2}\cdot \sum_{e\in
  p_j}\frac{x_{e}^{(\TSJ)}}{c_e}$.  Namely, each time a request is
rerouted, the weight of the path is at least halved.
Note that the halving is with respect to the weight of the path at the time it was allocated.

Let us consider request $r_j$. Let $p^* \eqdf \text{argmin}_{p\in \Gamma_j}\{\sum_{e\in p}\frac{1}{c_e}\}$.
By the choice of a ``lightest'' path and by Lemma~\ref{lemma:x ub}, the weight of path $p_j$ is upper bounded by
\[
\sum_{e\in p_j} \frac{x_e}{c_e} \leq \sum_{e\in p^*} \frac{x_e}{c_e} \leq 3 \cdot \sum_{e\in p^*} \frac{1}{c_e}\:.
\]

By Lemma~\ref{lemma:valx}, $x_e \geq 1/(4m)$, hence the weight of path $p_j$ is lower bounded by
\[
\sum_{e\in p} \frac{x_e}{c_e} \geq \frac {1}{4m} \cdot \sum_{e\in p} \frac{1}{c_e} \geq \frac {1}{4m} \cdot\sum_{e\in p^*} \frac{1}{c_e}.
\]

It follows that the number of reroutes each request undergoes is bounded by
$\log_2 \left( 12 m\right)$,
and the second part of the theorem follows.
\end{proof}
\begin{rem}
  Note that the first routing request will not be rerouted at all, the second routing request will be rerouted at most twice, and so on.
  In general, a routing request that arrives at time $t$ will be rerouted at most $|\alive_t|$ times.
\end{rem}

\section{Discussion}
We present a primal-dual analysis of an online algorithm in a
nonmonotone setting. Specifically, we analyze the online algorithm by
Awerbuch et al.~\cite{awerbuch2001competitive} for online routing of
virtual circuits with unknown durations.  We think that the main
advantage of this analysis is that it provides an alternative
explanation to the stability condition for rerouting that appears
in~\cite{awerbuch2001competitive}.  According to the primal-dual
analysis, rerouting is used simply to preserve the feasibility of the
solution of the covering LP.

Our analysis provides a small improvement compared
to~\cite{awerbuch2001competitive} in the following sense.  The optimal
solution in our analysis is both totally flexible (i.e., may
reroute every request in every time step) and splittable (i.e., may
serve a request using a convex combination of paths). The optimal
solution in the analysis of Awerbuch et
al.~\cite{awerbuch2001competitive} is only totally flexible and must
allocate a path to each request.

The primal-dual approach of Buchbinder and Naor~\cite{BNsurvey} is
based on bounding the change in the value of the primal solution by
the change in the dual solution (this is often denoted by $\Delta P
\leq \Delta D$). The main technical challenge we encountered was that
this bound simply does not hold in our case. Instead, we use an
averaging argument to prove an analogous result (see Lemma~\ref{eqn:np}).

\bibliographystyle{alpha}
\bibliography{reroute}

\end{document}